\documentclass{article}
\usepackage[papersize={8.5in,11in},margin=1in]{geometry}

\usepackage[utf8]{inputenc} 
\usepackage[T1]{fontenc}    
\usepackage{hyperref}       
\usepackage{url}            
\usepackage{booktabs}       
\usepackage{amsfonts}       
\usepackage{nicefrac}       
\usepackage{microtype}      

\usepackage{caption}
\usepackage{subcaption}
\usepackage{xcolor}
\usepackage{algorithm}
\usepackage{algorithmic}
\usepackage{amssymb, amsmath, latexsym, amsfonts, graphicx, amsthm}
\usepackage[numbers]{natbib}

\usepackage{bm}

\usepackage{tikz}
\usetikzlibrary{positioning}
\usetikzlibrary{shapes,arrows}
\usepackage{url}

\newtheorem{example}{Example}
\newtheorem{definition}{Definition}
\newtheorem{Lemma}{Lemma}
\newtheorem{theorem}{Theorem}

\title{A Game-theoretical Approach to Analyze Film Release Time}
\author{Hua Liu \\
ITCS, Shanghai University of Finance and Economics\\
\texttt{liu.hua@mail.shufe.edu.cn}
\and Mengjing Chen\\
IIIS, Tsinghua University\\
\texttt{cmj16@mails.tsinghua.edu.cn}
\and Xiaolong Wang\\
Jiuquan Satellite Launch Center\\
\texttt{379139010@qq.com}
\and Zihe Wang\\
ITCS, Shanghai University of Finance and Economics\\
\texttt{wang.zihe@mail.shufe.edu.cn}
}

\date{}

\begin{document}
\maketitle
\begin{abstract}
  Film release dates play an important part in box office revenues because of the facts of obvious seasonality demand in the film industry and severe competition among films shown at the same time. In this paper,  we study how film studios choose release time for movies they produce to maximize their box offices. We first formalize this problem as an attraction competition game where players (film studios) consider both potential profits and competitors' choices when deciding the release time. Then we prove that there always exists a pure Nash equilibrium and give the sufficient condition of the uniqueness of the Nash equilibrium. Our model can be generalized to an extensive game and we compute the subgame-perfect equilibrium for homogeneous players. For the case that one film studio could have multiple movies to release, we prove that finding a player's best response is NP-hard and it does not guarantee the existence of a pure Nash equilibrium. Experiments are provided to support the soundness of our model. In the final state, most of film studios, accounting for 84 percent of the market, would not change their release time. The behaviors of film studios imply they are following some strategies to reach a Nash equilibrium.

\end{abstract}

\section{Introduction}\label{sec:intro}
The global market of the film industry is 41.7 billion dollars in 2018 \cite{boxofficerev}. 
It keeps steadily growing at the speed of three percent annually in the last six years.
As a key component in the entertainment industry, the sale of a movie has been extensively studied in the academic field 
\cite{sivaprasad2018multimodal,ruhrlander2018improving,geng2015pre,dellarocas2007exploring,ghiassi2015pre}.
A bunch of factors can influence the sales of a film, for instance, 
the director, the film starring and the story.
Besides these obvious elements which are related to the quality of a movie, 
film's release date is one of the last major vehicles by which studios compete with each other \cite{Einav2010NOT,feng2017quality}.
                                
In the film industry, the release date of a film has a significant effect on the total box office of the film \cite{chiou2008timing}. 
In fact, the first week after release usually accounts for a great proportion in the total box office of a film, which makes the film's box office is time-sensitive. In this paper, we assume every film is shown for only one week. It means the gross sale of a film is only from the first releasing week. 

Films that are showing will compete with each other on their quality and popularity. Every film prefers to the time slot with a low level of competition intensity. For example, a middle-budget flick is less likely to choose a release date when there is a blockbuster film being shown since the blockbuster film will attract the most audience and occupy a large number of screens on cinemas \cite{krider1998competitive,sochay1994predicting}. 
Also, blockbuster films would like to avoid releasing at the same time in the case of box office dilution \cite{einav2007seasonality}. 

On the other hand, a film's sale is related to the seasonality in demand. A film is usually released on major holidays and Fridays rather than normal workdays even though the competition is more severe for the reason that major holidays imply underlying huge demands \cite{Einav2010NOT}. 
A film avoids contending the audience with other big events like the football World Cup Final which results a low demand.

A film studio seeks to maximize its box office which requires a comprehensive consideration of two main reasons (competition and demand) when choosing the release date. The price competition among film tickets is ignored since ticket prices are very close during the same period and cause a limited impact on box office \cite{orbach2007uniform}.

Due to the development of the information society, film studios are able to acquire the information about the film industry and competitors, assisting them to make wiser date release decisions.  For the film industry information, one can predict the demand (number of audience that go to cinemas) by observing the box office and the number of ticket sales in history \cite{einav2007seasonality,terry2011determinants}. For film information, whether a film is of high-quality and attractive is easy to judge by information of the film like director, starring, budget, trailers and so on \cite{redondo2010modeling,ainslie2005modeling}. 
Film studios can also get a movie's information by a sampling survey after its advanced screening.
Hence we regard the market size information and film popularity information as common knowledge in our model. 

Beyond the film release application, our framework and results are applicable to many similar scenarios.
For example, the live streaming video market grows in an amazing speed which is valued at $\$30$ billion in 2016 \cite{livevideo}. 
In a live streaming platform, one network anchor takes mainly the demand and competition into consideration 
when choosing the broadcasting time.

\subsection{Our contribution}
We analyze the time release problem in the film industry and our contributions can be summarized as follows.
\begin{itemize}
	\item First and foremost, we formalize the film release problem as a theoretical attraction competition game problem based on the reality. To the best of our knowledge, this is the first attempt to analyze film release problem in a game-theoretic aspect with a specific utility model. 
 	\item In the normal form setting, we prove that there always exists a pure Nash equilibrium. To reach a Nash equilibrium, players can run a greedy algorithm to select time slots sequentially in the decreasing order of the popularity degree. This result is consistent with the fact that blockbuster films usually decide the release time earlier than small-budget flicks.
Experiment shows that film studios that account for 84 percent of the market make right choices on release time.

	\item We generalize our model to an extensive game where players choose time slots in an arbitrary order. We compute the subgame-perfect equilibrium when players have the same popularity degree.  When there are two time slots, we prove that the outcome of subgame-perfect equilibrium forms a Nash equilibrium. 
	\item At last we consider the more general case where each film studio could have multiple films to release. In this setting we give a negative result that Nash equilibrium does not always exist. Furthermore, it is a NP-hard problem for a film studio to find a best response (a set of release dates) given other players' actions fixed.
\end{itemize}

\subsection{Related work}
The film industry is widely studied by researchers both from economics and computer science on a variety of subjects. Many researchers concentrate on box office forecasting, by using learning methods on consumers or news data \cite{mishne2006predicting,geng2015pre,ruhrlander2018improving}. However, they ignore the competition between film studios.
Other topic like revenue sharing \cite{einav2007seasonality,prag1994empirical,moretti2011social} in the film industry is also well-studied, which bring insight into characterizations of the market. 

Our attraction competition model is similar to the model of spacial competition game which have a large body of work \cite{hotelling1990stability,osborne1993candidate,sengupta2008hotelling,rodrigues2017non,anastasiadis2018heterogeneous}. An instance of spacial competition games is that retailers choose the locations of shops to maximize their profits.
In a spacial competition game, a player's profit is defined to be the number of customers who choose the player where customers are distributed over all locations and usually go to the nearest shops \cite{hotelling1990stability}. 
The key difference between their models and ours is how players compete for resources (customers).
They assume customers are indifferent in  players (shops) with same distances to them so the resource is shared by players equally \cite{feldman2016variations,shen2017hotelling}. In our model the resource is proportional shared according to players' types (popularity degrees).

The most related works is \cite{Einav2010NOT}. 
He proposes an empirical model for film studio's utility. 
By assuming that film studios' actions form a Nash equilibrium in the real data,
he pins down a set of parameters in the empirical model.
The main difference is that the utility functions in \cite{Einav2010NOT} are not specified but with parameters learned from data while we use a concrete model which makes the model clear and theoretical results possible.

\section{Attraction Competition Model}\label{sec:mdoel}
In our model, we have $n$ film studios (players) who decide the film release time and each of them has only one film to release. Player $i$'s type is 
the film $i$'s popularity degree, denoted by $\theta_i$. 
All players' types are common knowledge.
There are $m$ available time slots for releasing, denoted by $M=[m]$,  from where film studios make a choice. Let player $i$'s choice(action) be $ a_i \in M$. 
For each time slot, we assume there is a fixed number of audience that movies compete to attract. 
The demand of the audience in $j$-th time slot is denoted by $d_j$ which is also public information.

Next we define the utility functions.
Player $i$'s utility $u_i(\bm{a},\bm{\theta})$ is defined to be the number of audience who watches player $i$'s film, where $\bm{a}=(a_1,a_2,\dots, a_n)$ and $\bm{\theta}=(\theta_1,\theta_2,\dots,\theta_n)$ are the action profile and the type profile of all players respectively.
Players are selfish and utility maximizers.
To simplify our model, we assume each film only gets audience from the very time slot when it is released.
The movies that released in the same time slot compete with each other in a way that 
they divide the audience population in proportional to the popularity degrees. We use $C_j$ to denote the set of players who choose time slot $j$. 
Then the utility of player $i$ can be formally expressed as
$$
u_i(\bm{a},\bm{\theta})=d_{a_i}\cdot \frac{\theta_i}{\sum_{k\in C_{a_i}} \theta_k},\quad \forall i\in [n]
$$
We use an illustrative example throughout this paper for better understanding.
\begin{example}
\label{3m:example}
There are 3 movies in total with popularity degree $\theta_1=4, \theta_2=3, \theta_3=2$ and two time slots with audience demand $d_1=12, d_2=9$.
Then every player's  action space is $\{1,2\}$.

Suppose all three film studios release at the first time slot. In this case, the three players' utilities would be
$u_1=12 \cdot \frac{4}{2+3+4}=5.33$,$u_2=12 \cdot \frac{3}{2+3+4}=4$,$u_3=12 \cdot \frac{2}{2+3+4}=2.67$.
\end{example}

\section{Nash Equilibrium}\label{sec:nash}
Given each player's utility expression based on their actions, the attraction competition model becomes a normal form game.
In this section, we study the existence of pure Nash equilibrium in this game.
Let the pure strategy $s_i : \Theta_i \rightarrow M$ be the mapping from every possible type player $i$ could have to the choice he would make if he is that type. We denote that $\bm{s}= (s_1,s_2,\dots,s_n)$ is the strategy profile of all players and the $\bm{s}_{-i}$ is the strategy profile of all players except for player $i$. We use the standard definition of Nash equilibrium.

\begin{definition}
When there is a strategy $\bm{s^*}$ that no player can increase his utility by changing releasing time unilaterally, we call it a Nash equilibrium, that is
\begin{eqnarray*}
\forall i, s_i:u_i(s^*_i, \bm{s^*}_{-i}) \geq u_i(s_i, \bm{s^*}_{-i}).
\end{eqnarray*}
\end{definition}
In Example \ref{3m:example}, the action profile that all players choose the first time slot is not a Nash equilibrium. It is easy to verify $C_1=\{1,3\}, C_2=\{2\}$ constitute a Nash equilibrium. 

One major question is whether there always exists a Nash equilibrium and how players achieve it.

\begin{theorem}
In the attraction competition game, there always exists a pure Nash equilibrium.
If players choose the release time greedily in the decreasing order of the popularity degree, the reached schedule is a Nash equilibrium.
\label{existence}
\end{theorem}
\begin{proof}
When player $i$ chooses time slot $j$, his utility is $\frac{d_j\cdot \theta_i}{\sum_{k\in C_j} \theta_k}$.
Since $\theta_i$ is a fixed parameter, player $i$ is maximizing $\frac{d_j}{\sum_{k\in C_j} \theta_k}$ which is equivalent to minimizing $\frac{\sum_{k\in C_j} \theta_k}{d_j}$. 

We reduce the problem to a selfish routing game where there is a network  consisting of $m$ parallel links connecting a source node and a destination node. Each of $n$ network users routes a particular amount of traffic, 
denoted by $d_i$ for user $i$, along a link.
When users in set $C_j$ choose the $j$-th link,  they have the same latency cost $\frac{\sum_{k\in C_j} \theta_k}{d_j}$ that they want to minimize. 

The attraction competition game have the same utility structure as this selfish routing game.
Therefore they have the same Nash equilibrium. According to Theorem 1 in \cite{fotakis2002structure}, there always exists a Nash equilibrium in the routing game. According to Theorem 2 in  \cite{fotakis2002structure}, a Nash equilibrium can be derived when users make greedy decisions in the decreasing order of the amount of traffic.
We conclude our theorem.
\end{proof}

The theorem says players can reach a Nash equilibrium in a quite efficient and easy way. In the data from real industry, the change of release time can be seen as a sign of better responding to other film studios’ actions. The low frequency of such behaviors supports that film studios run the greedy algorithm can reach a Nash equilibrium efficiently. More discussions will be found in the experiment.

We revisit Example \ref{3m:example}. Let's see what we get if all players adopt greedy algorithm sequentially.
Player 1 chooses first since he has the largest popularity degree, he chooses time slot 1 since $d_1> d_2$.
Then player 2 chooses time slot 2 since $d_1\cdot \frac{3}{4+3}<d_2\cdot \frac{3}{3}$.
At last player 3 chooses time slot 1 since $d_1\cdot \frac{2}{4+2}>d_2\cdot \frac{2}{3+2}$. 
The final state is $C_1=\{1,3\}$ and $C_2=\{2\}$ which is a Nash equilibrium.

\begin{theorem}
\label{homo}
If all movies have the same popularity degree, then the Nash equilibrium is unique.
\end{theorem}
All missing proofs are provided in the full version due to the lack of space.
\section{Extensive Form Game}\label{sec:extensive}
In this section, we generalize this problem to an extensive form game with perfect information which captures the sequence of the actions of the players. Players choose actions in the order of 1 to $n$ and they can see all the events that have previously occurred. Each player's pure strategy specifies an action to be taken at each of his nodes.
In this section, we consider the subgame-perfect equilibrium (SPE).
\begin{definition}[subgame-perfect equilibrium]
The subgame-perfect equilibria(SPE) of a game $G$ are all strategy profiles $\bm{s}$ such that for any subgame $G'$ of $G$, the restriction of $\bm{s}$ to $G'$ is a Nash equilibrium of $G'$.
\end{definition}

\subsection{Homogeneous players}
Players are called homogeneous(heterogenous) when their films have the same (different) popularity degrees. 
We first solve the subgame-perfect equilibrium (SPE) computation problem when players are homogeneous. 
In this case, the distribution of the number of players in time slots is unique by Theorem \ref{homo}. 
In the unique Nash equilibrium, w.l.o.g., we assume
the utility density in each time slot (defined as the demand divided by the number of films in a slot) is in decreasing order such that $d_1/|C_1|>d_2/|C_2|>...>d_m/|C_m|$. The case when there are same utility densities, i.e. $d_{i}/|C_i|=d_{i+1}/|C_{i+1}|$,  can be easily generalized.
For sake of convenience, we assume there is only one best response at each node. Thus there is an unique subgame-perfect equilibrium which can be computed by backward induction. 
\begin{theorem}
Following the action path chosen in the unique subgame-perfect equilibrium, player $\sum_{k<i} |C_k|+1$ to player $\sum_{k\leq i} |C_k|$ will choose the i-th time slot for any $i\in [m]$.
\label{thm_extensive}
\end{theorem}
This theorem does not give the explicit strategy on every node. 
It only says what happens on the path from the root node to the leaf node chosen in the SPE.
As the theorem implies, players with priorities to choose earlier earn higher utilities in the outcome of the SPE.

\subsection{Heterogenous players}
The analysis becomes complicated when players are heterogenous.
To simplify our problem, we compute SPE in the setting where  players make decisions in the decreasing order of popularization.
In the end of this section, we will discuss how the sequence of players impacts a player's utility.

In the extensive form game, 
sophisticated player will reason his impact on followers'  move instead of using greedy algorithm. 
Some players can get more utilities in subgame-perfect equilibrium.
When players adopt greedy algorithm in Example \ref{3m:example}.
Continue the computation in the last section, the players' utilities would be
$u_1= 12\cdot 4/(4+2)=8$, 
$u_2=9\cdot 3/3=9$, 
and $u_3=12\cdot 2/(4+2)=4$. 
When players engage in the extensive form game, Figure \ref{tree} shows SPE in Example \ref{3m:example}.
Players' utilities are $u_1=9, u_2=7.2, u_3=4.8$ where player 1 chooses the second time slot and players 2,3 choose the first time slot.

Usually, backward induction method is used to compute SPE. In our problem, the size of game tree is exponential, but the size of input, which includes the order of players and the demand of each time slot,  is polynomial. So backward induction method is exponential time-consuming. 
We have to figure out an efficient way to tackle this problem. 
One possible method is to use the technique which solves the problem in the homogeneous players case.
We first eliminate a set of impossible outcomes(leaves) that the path chosen in the SPE will not reach.
Then we solve the SPE in the reduced game tree instead and prove SPEs in two game trees are identical at the common nodes.

\begin{theorem}
\label{spe}
The outcome of subgame perfect equilibrium in extensive form game is a Nash equilibrium when there are two time slots.
\end{theorem}

To make the theorem clearer, we emphasize the difference between these two equilibria.
In the SPE, we care about every player's strategy
that consists of the moves on all his nodes in the game tree.
It is well known that these strategies constitute a Nash equilibrium.
But our theorem only focuses on the outcome of the SPE.
The theorem says the leaf node, achieved in the SPE, is a Nash equilibrium in the normal form setting where each player only takes one action .
It is interesting that subgame perfect equilibrium in our problem has such a property. 
In Example \ref{3m:example}, the outcome of SPE is 
$C_1=\{2,3\}, C_2=\{1\}$, which is indeed a Nash equilibrium. 

%
%
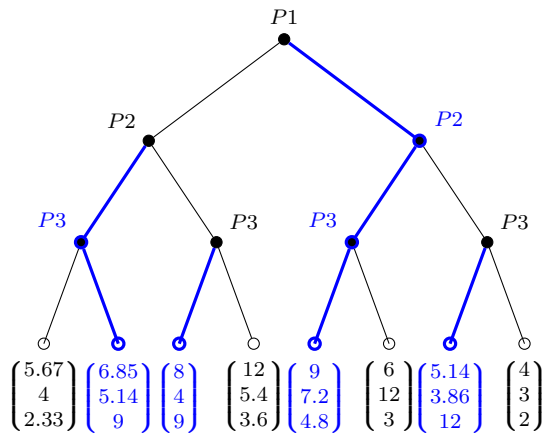
\begin{figure}
 	\begin{center}
    \small
 	 \begin{tikzpicture}[scale=.9,font=\footnotesize]
	   \tikzstyle{solid node}=[circle,draw,inner sep=1.5,fill=black]
	   \tikzstyle{hollow node}=[circle,draw,inner sep=1.5]
	   \tikzstyle{level 1}=[level distance=15mm,sibling distance=4 cm]
	   \tikzstyle{level 2}=[level distance=15mm,sibling distance=2.cm]
	   \tikzstyle{level 3}=[level distance=15mm,sibling distance=1.1cm]
	    \tikzstyle{level 4}=[level distance=15mm,sibling distance=1.5cm]
	   \node(0)[solid node,label=above:{$P1$}]{}
	   child{node[solid node,label=above left:{$P2$}]{}
	   	   child{[blue, very thick] node[solid node,label=above left:{$P3$}]{}
		   	   child{[black, thin] node[hollow node,label=below:{$ \left\lgroup\begin{matrix}5.67 \cr 4 \cr 2.33\end{matrix}\right\rgroup$}]{} }
			   child{[blue, very thick] node[hollow node,label=below:{$ \left\lgroup\begin{matrix}6.85 \cr 5.14 \cr 9\end{matrix}\right\rgroup$}]{} }
		   	}
		   child{node[solid node,label=above right:{$P3$}]{}
		   	   child{[blue, very thick] node[hollow node,label=below:{ $ \left\lgroup\begin{matrix}8 \cr 4 \cr 9\end{matrix}\right\rgroup$ }]{} }
			   child{node[hollow node,label=below:{$ \left\lgroup\begin{matrix}12 \cr 5.4 \cr 3.6\end{matrix}\right\rgroup$}]{} }
		   	}
	   	   }
	   child{ [blue, very thick] node[solid node,label=above right:{$P2$}]{}
	   	   child{[blue, very thick] node[solid node,label=above left:{$P3$}]{}
		           child{[blue, very thick] node[hollow node,label=below:{$ \left\lgroup\begin{matrix}9 \cr 7.2 \cr 4.8\end{matrix}\right\rgroup$ }]{} }
			   child{[black, thin] node[hollow node,label=below:{$ \left\lgroup\begin{matrix}6 \cr 12 \cr 3\end{matrix}\right\rgroup$}]{} }
		   	}
		   child{ [black, thin] node[solid node,label=above right:{$P3$}]{}
		   	   child{[blue, very thick] node[hollow node,label=below:{$ \left\lgroup\begin{matrix}5.14 \cr 3.86 \cr 12\end{matrix}\right\rgroup$}]{} }
			   child{ [black, thin] node[hollow node,label=below:{$ \left\lgroup\begin{matrix}4 \cr 3 \cr 2\end{matrix}\right\rgroup$}]{} }		   
		   	}
	   	};

          \end{tikzpicture}
     \end{center}
    \caption{Players' utilities are given in the bracket. The outgoing blue edges from  nodes indicate players' strategy. The outcome of the SPE is $C_1=\{2,3\}$, $C_2=\{1\}$ with utilities (9,7.2,4.8).}
\label{tree}
 \end{figure} 

We consider the reduced game tree whose nodes are all Nash equilibria. We use backward induction method to compute SPE.
If there is a polynomial number of Nash equilibria, then we can efficiently solve the problem.

However, this method fails when there is a superpolynomial number of Nash equilibria. 
We don't know if Theorem \ref{spe} will hold when there are more than two time slots neither.
More generally, it remains open that compute the subgame-perfect equilibrium in polynomial time.
The SPE computation problem may not lie in NP class.

Now we consider the decision order's impact on a player's utility. 
When players are homogeneous, we have seen that players who choose first have advantages, i.e., the utilities in the outcome will be higher. 
So players will fight to take the good positions first. 
There also exist cases when players are heterogeneous with large gap between the popularity degrees, 
weak players will compete to choose first and might get larger utilities than the strong players get eventually.
For example, when there are two time slots, $d_1=1+1/t-\epsilon$, $d_2=1$, $\theta_1=t$, $\theta_2=1$. If player 2 succeeds to choose first time slot before player 1. His utility $u_2$ would be $1+1/t-\epsilon$, larger than $u_1=1$.
We wonder if this phenomenon always holds for the heterogeneous players case.
Since players can make decisions at any time in reality, does the player has an incentive to choose time slot in an earlier time? Is it always better to choose earlier than later?
Formally, the open question is when should a player make a decision in an extensive form game assuming that the other players' action sequence is fixed.

\section{Two Variant Models}\label{sec:group}
\subsection{Each film studio has multiple movies to release}
Each film studio has to decide the time slot for every movie it has.
We first prove that the complexity of calculating a player's best response is NP-hard. Then we take an example to show that the pure Nash equilibrium of the attraction competition game may not even exist in this setting.

\begin{theorem}
	The problem of finding player $i$'s best response of the attraction competition game in this setting is NP-hard.
\end{theorem}

\begin{theorem}\label{thm:no_ne}
	A pure Nash equilibrium of the attraction competition game may not exist when a film studio has multiple films to release.
\end{theorem}
We prove the Theorem \ref{thm:no_ne} by giving an example.
\begin{example}
	There are 2 players $A$, $B$ where player $A$ has two movies $A_1$ and $A_2$ with popularity degree $1$, $2$ respectively and player $B$ has one movie $B_1$ with popularity degree $1$. There are two time slots with same demand $d_1=d_2=1$. Every player's action space is $\{1,2\}$.	
\end{example}
	Player $B$ always tries to choose the same slot where $A_1$ released.
	Player $A$ wants $A_2$ released in the same time slot as $B_1$ and $A_2$ released in a different time slot. 
		Hence there is no pure Nash equilibrium in the game.
	
\subsection{Each movie lasts for two slots}
\label{2slots}
We consider the setting that each movie gets utility from both the time slot it releases and the next one. 
In this subsection, we assume a film's popularity degree does not decrease over time.
Formally, we define $$
u^*_i(\bm{a},\bm{\theta})=d_{a_i}\cdot \frac{\theta_i}{\sum_{k\in C_{a_i}+C_{a_i-1}} \theta_k}
+d_{a_i+1}\cdot \frac{\theta_i}{\sum_{k\in C_{a_i+1}+C_{a_i}} \theta_k}
$$
For this problem, we only solve the case for homogeneous players.
\begin{theorem}
	There always exists a Nash equilibrium.
\end{theorem}
This new game still belongs to the congestion game. But traditional method does not work for heterogeneous players any more. To appreciate the difficulty of the challenge, a player's utility depends on not only the players who choose the same time slot but also the players who choose the neighbor slots. 
It becomes very complicated and difficult to measure the impact of a player's best response. 
\section{Inferring from data}\label{sec:experiment}
We design a sequence of experiments to demonstrate that our model makes sense for the real movie industry and strategies of film studios in the real world form a Nash equilibrium. In this section, we will show that: (i) the movie's box office revenues is time-sensitive and highly affected by the release time; (ii) the number of audience is proportional to the popularity degrees of movies from the data; (iii) film studios in the real industry always give best response to other competitors' actions; (iv) the fact that film studios who have movies with higher box offices decide the release time earlier is a way that reach Nash equilibrium at last; (v) the final schedule of release dates in the real market is close to a Nash Equilibrium.

The data we use is from Maoyan \cite{maoyan}, a leading online movie tickets service in China. The dataset includes movies on the show in China by day for nearly 8 years(Jan. 1st, 2011 $\sim$ Oct. 30th, 2018) and consist of more than 68 thousand records. A record contains 24 attributes, including date, movie name, box office at this day, attendance rate, box office, release date and etc. The total number of movies in the dataset is 2818. Our experiments which involve the attendance rate only use the data after Jan. 12th, 2015, since there is no information of the attendance rate until 2015. The smaller dataset contains 40654 records and 1703 movies.

\subsection{Soundness of model}
\begin{figure}
	\centering
	\includegraphics[width=0.43\textwidth]{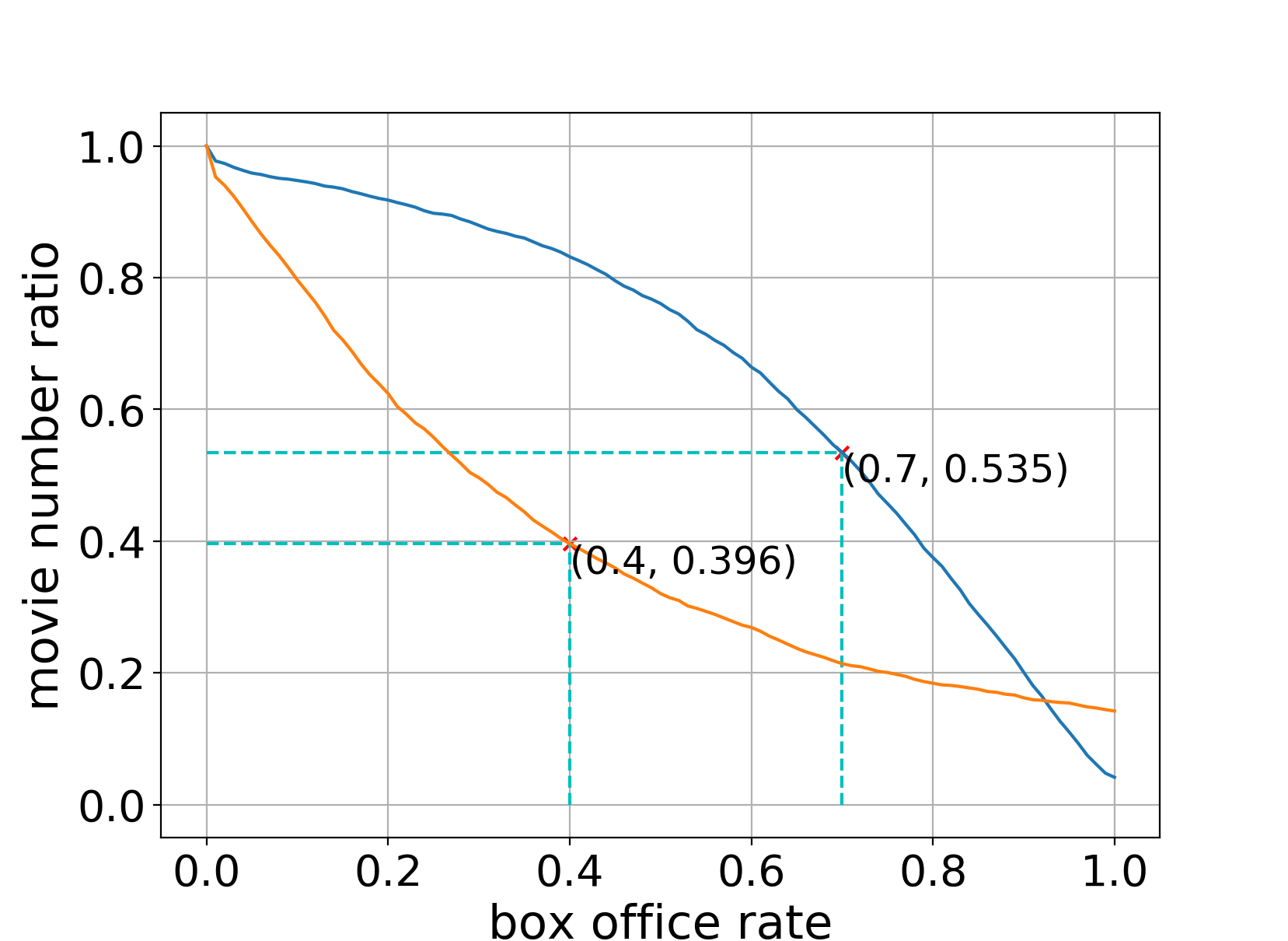}
	\caption{Box office rate at first/second week. A $(x, y)$ point in the blue line means that the ratio of movies whose {\em box offices of the first week} divided by {\em total box offices} are equal or greater than $x$ is $y$.  And in the orange line, a $(x, y)$ point represents that the ratio of movies whose {\em box offices of the second week} divided by {\em box offices of the first week} are equal or greater than $x$ is $y$.}
	\label{fig:box office rate}
\end{figure}

\paragraph{\textbf{The first week after release accounts for a great proportion in total box office of a movie.}}

We find that over half of movies gain more than $70\%$ of their final box offices at the first week. The first week is highly important to the box office (see the blue line in Figure \ref{fig:box office rate}). Moreover, the box office drops rapidly as the play time goes on. For about $60\%$ movies, the box offices at the second week is less than two-fifths of those at the first week (see the orange line in Figure  \ref{fig:box office rate}). All the facts above illustrate that the movie's box office is time-sensitive and the release time makes a significance to the final box office of a movie. Hence it is appropriate for us to focus on movies that released on the same day and consider their competitions at the first week.


\paragraph{\textbf{Movie's box office is proportional to the popular degree in a single time slot.}}
Movie theaters' behaviors have great impacts on a movie's box office. Just imagine that all theaters resist a movie for some special reasons and do not play it then this movie's box office will be zero. 
However, theaters are utility maximizers with a simple adjustment strategy which results a reasonable formula for a movie's box office.
Theaters dynamically adjust the proportion of screenings of movies to balance the attendance rates.
When a theater finds a movie has a relatively high (low) attendance rate, the theater will arrange more (less) screens for this movie on the next day. 
In fact, a theater computes every movie's profit in one screening.
Intuitively, the theater intends to arrange all screenings for the movie with the highest profit per screening for the time being. However the profit per screening of that movie will drop. 
A movie's profit in one screening is proportional to the number of audience, assuming the prices are same, the theater will arrange the number of screenings of one movie proportional to its demand, i.e., the popular degree. Hence each movie's box office is proportional to its popular degree.

Ideally, theaters should make the number of screenings in proportion to the number of audience, namely keep the attendance rates of all movies on the show balanced. 
This is a strong evidence to support our assumption that each player's utility is proportional to its popularity degree.

In this experiment, we only consider the `regular' movies whose attendance rate is above a threshold. Movies with low attendance rates usually have various reasons to stay on the screen, like contracts.
1311 days in total are divided into 264 periods which last about 6 to 7 days usually. We call a new period begins if the screenings of new movies on that day are larger than ten percent of all screenings. Thus in each period, movie competition is relatively stable. 
We use the following formula to judge if the theater is balancing the attendance rate:
$$\sum_i (\alpha_{i,t} -avg_t) (\beta_{i,t} - \beta_{i,t+1})$$
Here $\alpha_i$ means movie $i$'s attendance rate on day $t$,
$avg_t$ means the average attendance rate of regular movies on day $t$.
$\beta_{i,t}$ represents the screenings of movie $i$ on day $t$.
If $\alpha_{i,t}<avg_t$ which means the performance of movie $i$ is below the average, then the theater should decrease its number of screenings and get $\beta_{i,t}<\beta_{i,t+1}$.
Ideally, we should have $(\alpha_{i,t} -avg_t) (\beta_{i,t} - \beta_{i,t+1})$.
Therefore, we say the theaters' adjustment is rational if the summation is larger than zero.
The result is shown in Figure \ref{fig: attendance rate}. We can see theaters are rational in most of days.
%
%
\begin{figure}
    \centering
    \includegraphics[width=0.52\textwidth]{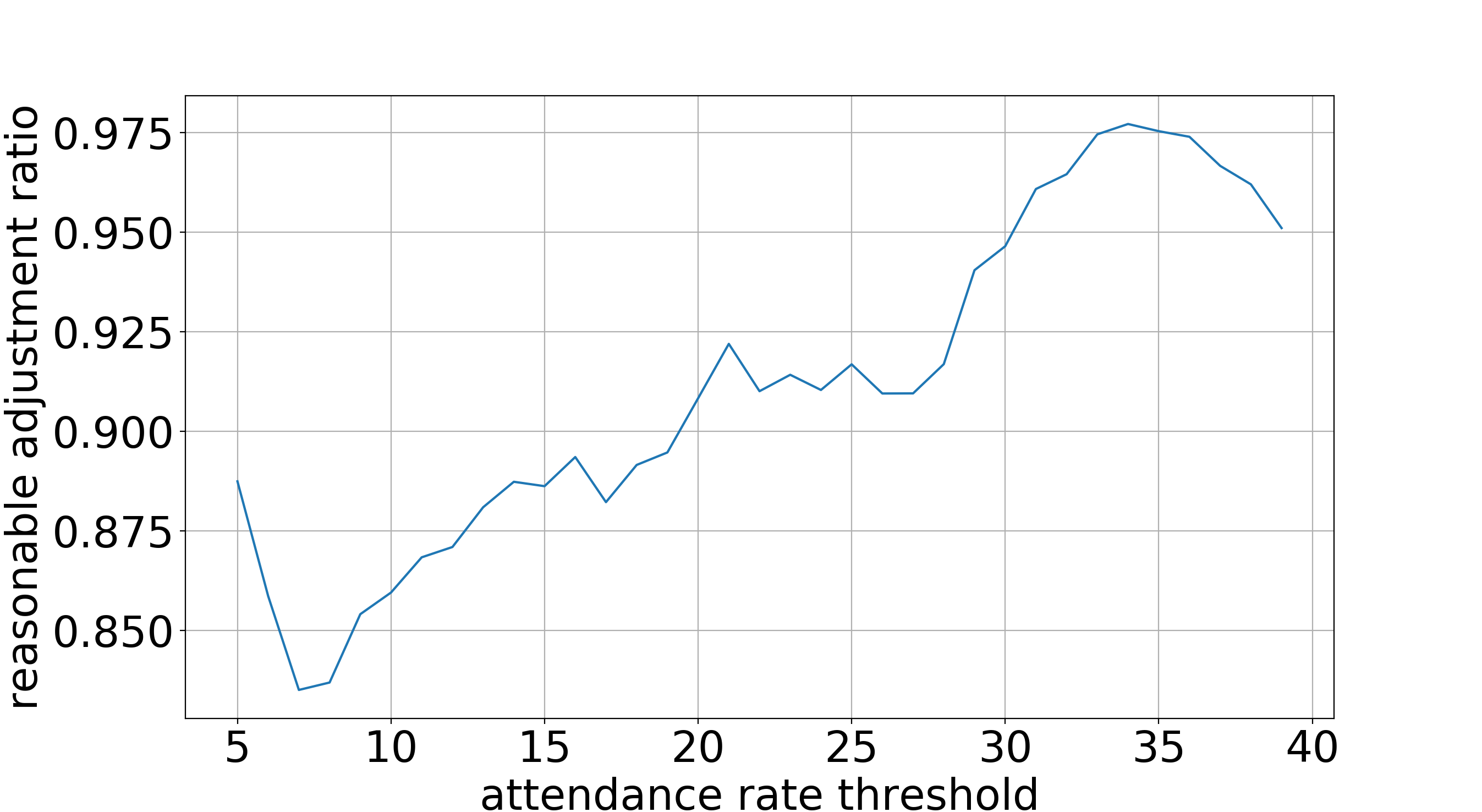}
    \caption{Theaters' rational adjustment ratio.}
    \label{fig: attendance rate}
\end{figure}
 
\paragraph{\textbf{Film studios change the release time when some foreign movies happen to released at the same time.}}
Foreign movies played in China are usually big-budget and star-studded movies or have good word-of-mouth effects, which means foreign movies are likely to be highly popular. Hence a rational domestic movie will avoid being released at the same time as a foreign movie released in case of box office dilution. Foreign movies usually have limitations on choices 
because of the film import formalities. The fact is that the release dates of many domestic movies are changed when foreign movies happen to be released at the same time. For example, three star-studded domestic movies have been set Aug. 25th as release date several months before that day. However, {\em Valerian and the City of a Thousand Planets} and {\em Cars 3} made the decision to be on the show on that day latter. Then the release dates of all three movies are changed after that news. This phenomenon demonstrates that film studios in the real industry always give best response to other competitors’ actions.
\begin{table}[h]
	 \centering
\begin{tabular}{|r|r|}
	\hline
	\textbf{box office(million)} & \textbf{average days in advance} \\ \hline
		$[100, \infty)$                   &             97.61            \\ \hline
 		$[10, 100)$                           &         49.72                \\ \hline
 		$[1, 10)$                           &              32.19           \\ \hline
		$[0, 1)$                           &              25.60           \\ \hline
	\end{tabular}
	\caption{The average days in advance of movies. }
	\label{table:days}   
\end{table}

{\em \textbf{Movies with higher box offices are made release decisions earlier.}} 
In the real industry, movies with high probability of attracting a large audience are decided release time early. 
In the last year, there are 130 movies made in China having detail information including the box office, release date and the time stamp of decision making.
These movies fall into four categories based on their total box offices shown in Table \ref{table:days}.
We compute the average number of days that when the decision is made before release date for each movie.
Movies with a famous director and popular movie stars are more promising to have good quality and popularity among the audience. These movies usually are made release decisions one year or half before release, which is much earlier than other movies. The higher box office, the earlier time film studios make decisions. 
This phenomenon is consistent with the case that players make plans in the decreasing order of popularity degree. In such a way, players reach a Nash equilibrium easier and avoid making deviations frequently. 

\paragraph{\textbf{The schedule of release dates in the real market is close to a Nash Equilibrium.}} We design experiments to test whether the release dates of most movies are best responses to other competitors' actions, by checking whether film studios gain less box offices after shifting their release dates. We first take the total box office in several consecutive days as the market size of a slot and infer the popularity degree of each movie from the box office . For each movie, we compare the real box office with the expected box office of the movie if its release time is changed unilaterally.

We partition days into slots.
Since movies are usually released on Friday, most of slots are defined as a sequence of days that begins from Friday and end to the next Thursday. Slots are adjusted slightly when there are major holidays, like Spring Festival, the National Day, New Year and so on.

The  popularity degrees are inferred from box offices. We normalize the sum of popularity degrees in the first slot, and the popularity degree of each movie released on the first slot equals movie's box office in such slot divided by the slot's total box office. There exists movies released in previous slots but still on show in the present slot. We call such movies as {\em old movies} in the present slot and call movies first released in the present slot as {\em new movies}. \footnote{We have shown in Section \ref{2slots} that it is difficult to derive a theoretical result when a movie can get box office in two slots. In practice, we have to deal with the competition from {\em old movies}.}
A discount factor $\gamma$ is used to reflect the decrease of the popularity degrees of old movies. In slot $k$ (except the first slot), each new movie's popularity degree is calculated as follows:
\begin{gather*}
\theta_i = \frac{B_{i,k}}{\sum_{j \in OM_k}B_{j,k}}\cdot \sum_{j \in OM_k} \gamma^{n_{j,k}} \theta_j  
\end{gather*}
where $B_{ik}$ denotes movie $i$'s box office in slot $k$, $OM_k$ denotes the set of old movies in slot $k$ and $n_{j,k}$ represents how many slots movie $j$ have been on show until slot $k$.

We suppose that film studios have the ability to shift their actual release slots by no more than 4 slots. If the expected box office of a movie after the release slot changed is lower than 1.1 times the actual box office, we think the actual release slot is the best responses of this film studio to other competitors' strategies. We use the  {\em best-response rate} to represent to what extent players give best response, which has the formulation:
\begin{gather*}
\text{best-response rate} = \frac{\sum_{i \in BST} B_{ir(i)}}{\sum_{i \in N} B_{ir(i)}}
\end{gather*}
where $r(i)$ denotes movie $i$'s actual release slot is $r(i)$, and $BST$ and $N$ represent the set movies with best responses and the total set of movies respectively. Obviously, the best-response rate equals 1 if all movies are best responses. The best-response rates of shifting $x$ slots ($x \in [-4,+4] $) are shown in Table \ref{table:ne}. The results strongly prove that the schedule of release dates in the real market is a Nash Equilibrium.

\begin{table}[h]
	 \centering
\begin{tabular}{|r|r|r|r|r|}
	\hline
	\textbf{shift $x$ slots}   & -4 & -3 & -2 & -1   \\ \hline
	\textbf{best-respnse rate} & 0.8624 & 0.8769 & 0.8918 & 0.9083   \\ \hline
	\textbf{shift $x$ slots}   & 1 & 2 & 3 & 4  \\ \hline
	\textbf{best-respnse rate} & 0.8413 & 0.8858 & 0.8982 & 0.8629   \\ \hline

 	\end{tabular}
	\caption{The best-response rates of shifting slots. }
	\label{table:ne}   
\end{table}
If a film studio chooses a releasing time randomly, then a movie has a half probability to increase profit by deviation to another slot and the number in the table will be around 0.5.

\clearpage
\bibliographystyle{plainnat}
\medskip

\bibliography{ref}


\end{document}